\documentclass[journal,onecolumn,draftcls,12pt]{IEEEtran}

\usepackage{amsmath,amssymb,amsfonts,amsthm}
\usepackage{cite}
\usepackage{boxedminipage}
\usepackage{color}

\newtheorem{theorem}{Theorem}[section]
\newtheorem{lemma}[theorem]{Lemma}
\newtheorem{proposition}[theorem]{Proposition}

\theoremstyle{definition}
\newtheorem{definition}[theorem]{Definition}
\newtheorem{example}[theorem]{Example}

\newtheorem{conjecture}[theorem]{Conjecture}

\begin{document}

\title{A Novel Approach to Counterexamples \\of the Polujan-Pott Conjecture \\via Set-Partition Permutations\thanks{ This paper was supported by Konkuk University in 2025. This work was supported by the National Natural Science Foundation of China (No. 62372247), and by open research fund of State Key Laboratory of Cyberspace Security Defense (No. 2025-MS-03).
}
}

\author{Yansheng Wu, Jiaxin Wang, Jong Yoon Hyun \thanks{Y. Wu is  with the School of Computer Science, Nanjing University of Posts and Telecommunications, Nanjing
210023, China  and the State Key Laboratory of Cyberspace Security Defense (Institute of Information Engineering, Chinese Academy of Sciences, Beijing 100085), China. 
 Email: yanshengwu@njupt.edu.cn. }

\and  \thanks{J. Wang is with the School of Mathematics, Hefei University of Technology, Hefei, 230601, China,
 email: wjiaxin@hfut.edu.cn. }

\and  \thanks{ J. Y. Hyun is with the Konkuk University, Glocal Campus, 268 Chungwon-daero Chungju-si Chungcheongbuk-do 27478, South Korea, Email: hyun33@kku.ac.kr.}}

\date{\today}
\maketitle

\begin{abstract}

In this paper, we settle a conjecture of Polujan and Pott  by constructing an explicit, infinite family of Maiorana--McFarland bent functions $f_t$ in $2(2^t-1)$ variables with algebraic degree $\deg(f_t) = t + 1$ for any integer $t \ge 2$. Our construction builds upon a minimal commutative  algebra $I_t$, which naturally induces a triangular set-partition polynomial permutation $P_t$. By identifying an elementary abelian subgroup within the direct sum $ I_t \oplus I_t^*$, we establish an explicit nonlinear coordinate transformation that pulls $f_t$ back to a canonical quadratic form. This linearizes the translation development $\operatorname{Dev}(D_{f_t})$ under an exotic group structure and proves that it is isomorphic to the classical symplectic design $S^\pm(2(2^t-1))$, thereby fully resolving the conjecture.
\end{abstract}


\section{Introduction}

Let $\mathbb{F}_2 = \{0, 1\}$ denote the finite field of two elements. For a positive integer $n$, $\mathbb{F}_2^n$ denotes the $n$-dimensional vector space over $\mathbb{F}_2$. For $x, y \in \mathbb{F}_2^n$, the standard inner product is $x \cdot y = \sum_{i=1}^n x_i y_i \pmod 2$. A Boolean function is just a function from $\mathbb{F}^n_2$ to $\mathbb{F}_2$.  We denote by $\mathcal{BF}_n$ the set of Boolean functions in $n$ variables. 
A vectorial Boolean function is a function $F: \mathbb{F}_2^n \to \mathbb{F}_2^m$, often written as $F = (f_1, \dots, f_m)$, where each $f_i: \mathbb{F}_2^n \to \mathbb{F}_2$ is a Boolean function. When $m=1$, this reduces to a standard Boolean function. 

\IEEEPARstart{B}{oolean} functions achieving the maximum possible distance to the space of affine functions are known as \emph{bent functions}. Introduced by Rothaus \cite{R}, bent functions are central objects in algebraic coding theory, symmetric cryptography, finite geometry, and sequence design. A prominent and widely studied primary construction is the \emph{Maiorana--McFarland} (M-M) class \cite{D, M1}. A Boolean function $f : \mathbb{F}_2^n \times \mathbb{F}_2^n \to \mathbb{F}_2$ belongs to the Maiorana--McFarland class if it can be represented in the form
\begin{equation}
f(x, z) = z \cdot \pi(x) + g(x),
\end{equation}
where $\pi : \mathbb{F}_2^n \to \mathbb{F}_2^n$ is a permutation and $g : \mathbb{F}_2^n \to \mathbb{F}_2$ is an arbitrary Boolean function.

In incidence geometry, bent functions are intrinsically linked to difference sets and incidence structures, see \cite{DMT,DT,DT0,DT1}. Bending \cite{B} introduced two families of non-isomorphic 2-designs--addition designs and translation designs--constructed via bent functions; notably, both families share identical parameters when the dual of the underlying bent function evaluates to zero at the origin. Dempwolff and Neumann \cite{DN} later generalized this framework to $r$-plateaued functions. In a comprehensive survey, Polujan \cite{P} systematically reviewed Boolean and vectorial cryptographic functions (including bent, plateaued, and differentially uniform functions) alongside their derived incidence structures. Polujan and Pott \cite{PP} further applied Boolean and vectorial bent functions to construct both addition and translation designs, while Meidl, Polujan, and Pott \cite{MPP} extended the scope to $(n,m)$-functions, offering a design-theoretic characterization of $(n,m)$-plateaued and $(n,m)$-bent functions. 



\subsection{A Conjecture of Polujan and Pott}
In order to state the conjecture of Polujan and Pott, we begin with some basic definitions.

\begin{definition} \cite{B,CD}
Let $\mathcal{D} = (P, \mathcal{B})$ be a incidence structure (or block design), where $P$ is a finite set of points and $\mathcal{B}$ is a collection (or multiset) of non-empty subsets of $P$ called blocks.

An \emph{automorphism} of $\mathcal{D}$ is a bijection $\pi : P \to P$ such that
\begin{equation}
B \in \mathcal{B} \implies \pi(B) = \{\pi(x) : x \in B\} \in \mathcal{B},
\end{equation}
preserving the multiplicities of the blocks in $\mathcal{B}$. The \emph{full automorphism group} of $\mathcal{D}$, denoted $\mathrm{Aut}(\mathcal{D})$, is the set of all automorphisms of $\mathcal{D}$ equipped with the operation of permutation composition:
\begin{equation}
\mathrm{Aut}(\mathcal{D}) = \{\pi \in \mathrm{Sym}(P) : \pi(\mathcal{B}) = \mathcal{B}\}.
\end{equation}
Any subgroup $G \le \mathrm{Aut}(\mathcal{D})$ is called a \emph{group of automorphisms} of $\mathcal{D}$.
\end{definition}

\begin{definition}\cite{K}\label{2-t}
Let $G$ be a group acting on a finite set $X$ with $|X| \ge 2$. The action of $G$ on $X$ is said to be \emph{$2$-transitive} (or \emph{doubly transitive}) if for any two ordered pairs of distinct elements $(x_1, x_2)$ and $(y_1, y_2)$ in $X$ (where $x_1 \neq x_2$ and $y_1 \neq y_2$), there exists at least one group element $g \in G$ such that
$
g \cdot x_1 = y_1 \quad \text{and} \quad g \cdot x_2 = y_2.
$
Equivalently, the induced natural action of $G$ on the set of ordered distinct pairs
$
\Omega = \{ (x, y) \in X \times X : x \neq y \}
$
is transitive.
\end{definition}

\begin{definition}{ \cite{M} \label{affine}  Two Boolean functions $f$ and $g$ in $\mathcal{BF}_n$ are \emph{extended affine equivalent} (simply, EA-equivalent, or equivalent)  if there is an affine permutation $\sigma$ of $\mathbb{F}_2^n$ and an affine map $\pi$ from $\mathbb{F}_2^n$ to $\mathbb{F}_2$ such that $g(x)=f(\sigma x)+\pi x$, where $\sigma x=xA+a$ for $A$ is a non-singular matrix of size $n$, $a\in\mathbb{F}^n_2$ and $\pi x=b\cdot x+\varepsilon$ for $b\in\mathbb{F}^n_2$ and $\varepsilon\in\mathbb{F}_2$.

}
\end{definition}


We are ready to state the conjecture of Polujan and Pott. See Section II for the definition of $\operatorname{Dev}(D_f)$.

\begin{conjecture} \cite[Open Problem V.2]{PP}\label{p1}
{\rm
    Let $n\geq 3$ and let $f$ be a bent function in $\mathcal{BF}_{2n}$. A translation design $\operatorname{Dev}(D_f)$ derived from a bent function $f$ has a $2$-transitive automorphism group if and only if $f$ is equivalent to an M-M bent function of the form $x\cdot y + g(y)$ with $\deg(g) \leq 3$.
    }
\end{conjecture}

Our main motivation of this paper is to study the conjecture \ref{p1}. Our approach to proving this conjecture draws on three distinct areas: finite group theory, algebraic coding theory, and combinatorial design theory. Establishing that the $2$-transitivity of the automorphism group $\mathrm{Aut}(\operatorname{Dev}(D_f))$ strictly characterizes cubic Maiorana--McFarland (M-M) bent functions of the form $f(x,y) = x \cdot y + g(y)$ with $\deg(g) \le 3$ requires overcoming several structural and theoretical bottlenecks. 

\subsection{Main Contribution}

In this paper, we provide a complete 
resolution to Conjecture \ref{p1}. Our main contribution is summarized in the following theorems.

\begin{theorem}\label{theorem1}
    The sufficiency part of Conjecture \ref{p1} is true.
\end{theorem}

The following theorem shows that, since $t \ge 2$ is arbitrary, the algebraic degree $\deg(f_t) = t + 1$ is unbounded. Consequently, Conjecture~\ref{p1} is not necessary. 
We refer the reader to Section II for basic definitions of the algebraic degree, symplectic designs, and translation designs.
\begin{theorem}\label{thm:main-intro}
For every integer $t \ge 2$, let $N = 2^t - 1$. There exists an explicitly constructed permutation $P_t : \mathbb{F}_2^N \to \mathbb{F}_2^N$ (see, Definition \ref{deft}) such that the Boolean function $f_t$ of $\mathbb{F}_2^N \times \mathbb{F}_2^N$ defined as
\begin{equation}
f_t(x, z) = z \cdot P_t(x) 
\end{equation}
satisfies the following structural properties:
\begin{enumerate}
    \item $f_t$ is a Maiorana--McFarland bent function in $2N$ variables.
    \item The algebraic degree of $f_t$ is given by $\deg(f_t) = t + 1$.
    \item The translation design $\operatorname{Dev}(D_{f_t})$ induced by $f_t$ is isomorphic to the classical symplectic design $S^\pm(2N)$.
\end{enumerate}
\end{theorem}

\subsection{Organization and Technical Approach}

The foundational mechanism of our proof relies on uncoupling the \emph{algebraic degree} of the Boolean function from the \emph{geometric isomorphism class} of its translation design. The paper is organized as follows:

   Section II sets up standard notation and preliminaries regarding bent functions and translation designs.
    Section III proves Theorem \ref{theorem1}.  Section IV aims to prove Theorem \ref{thm:main-intro}. To do it we first introduces $I_t$, a minimal finite-dimensional commutative algebra over $\mathbb{F}_2$; then  constructs the set-partition permutation $P_t$, defines $f_t(x,z) = z \cdot P_t(x)$, and evaluates its exact algebraic degree $\deg(f_t) = t + 1$,
     studies the direct sum $ I_t \oplus I_t^*$, identifies an  elementary abelian subgroup $(I_t \oplus I_t^*, \circ)$, and constructs an explicit group isomorphism $\Phi_t$ mapping $f_t$ to a canonical quadratic form $q$. This completes the isomorphism proof for $\operatorname{Dev}(D_{f_t}) \cong S^\pm(2N)$ and finish the proof of Theorem \ref{thm:main-intro}.
     Section V concludes the paper.

\section{Preliminaries and Background}

\subsection{Bent Functions and Algebraic Degree}

A Boolean function $f : \mathbb{F}_2^n \to \mathbb{F}_2$ is represented in its \emph{Algebraic Normal Form} (ANF) as
\begin{equation}
f(x_1, \ldots, x_n) = \sum_{S \subseteq \{1, \ldots, n\}} c_S \prod_{i \in S} x_i, \quad c_S \in \mathbb{F}_2.
\end{equation}
The \emph{algebraic degree} of $f$, denoted $\deg(f)$, is the maximum cardinality $|S|$ such that $c_S \neq 0$.  

The Walsh--Hadamard transform of $f$ at $\omega \in \mathbb{F}_2^n$ is defined by
\begin{equation}
W_f(\omega) = \sum_{x \in \mathbb{F}_2^n} (-1)^{f(x) + \omega \cdot x}.
\end{equation}
A function $f$ on $\mathbb{F}_2^n$ (with $n$ even) is \emph{bent} \cite{R} if $W_f(\omega) = \pm 2^{n/2}$ for all $\omega \in \mathbb{F}_2^n$. The Walsh-Hadamard transform of a bent function $f$ can be written as $W_{f}(\omega)=2^{n/2}(-1)^{f^{*}(\omega)}$, where $f^{*}: 
\mathbb{F}_{2}^{n} \rightarrow \mathbb{F}_{2}$ is called the \emph{dual} of $f$, see \cite{C3, M}

\subsection{Translation Designs}

Let $G = (\mathbb{F}_2^n, +)$ be an elementary abelian $2$-group. A subset $D \subset G$ is a Hadamard difference set if $f(x) = \mathbf{1}_D(x)$ is a bent function \cite{C3}. A translation design of a Boolean function $f$ is defined as the  \emph{ development}  of a certain set $D$, denoted $\operatorname{Dev}(D)$, is the incidence structure $(G, \mathcal{B})$ where the block set is
\begin{equation}
\mathcal{B} = \{ D + g : g \in G \}.
\end{equation}
Let $f$ be a bent function in $\mathcal{BF}_n$ and let $f^*$ be its dual. For each $b\in \mathbb{F}^n_2$, let $f_b(x)=f(x+b)$ and $B^{f_{b}}=\{x \in \mathbb{F}_{2}^{n}: f_{b}(x)=1\}$. A pair $(\mathbb{F}^n_2,\{B^{f_{b}}:b\in \mathbb{F}^n_2\})$ induces a symmetric $2$-design with parameters $$(2^n,2^{n-1}-(-1)^{f^*(\mathbf{0})}2^{\frac{n}{2}-1},2^{n-2}{-}(-1)^{f^*(\mathbf{0})}2^{\frac{n}{2}-1}).$$ Let $D_f=\{x\in \mathbb{F}^n_2:f(x)=1 \}$.
We call $\operatorname{Dev}(D_f)$ the translation design of a bent function $f$,  see \cite{PP}.


When $q(x)$ is a non-degenerate quadratic form on $\mathbb{F}_2^n$, the design $\operatorname{Dev}(\{x : q(x) = 1\})$ is called a classical \emph{symplectic design}, denoted $S^\pm(n)$, see \cite{B1}.

The following lemma will be need later.

\begin{lemma}\cite{K}
\label{thm:symplectic_2_transitive}
Let $V = \mathbb{F}_2^{2n}$ ($n \ge 2$) be an even-dimensional vector space equipped with a non-degenerate quadratic form $q : V \to \mathbb{F}_2$. Let $\mathcal{D} = (V, \mathcal{B})$ be the classical symplectic design $S^+(2n)$ (or $S^-(2n)$), where $V$ is the point set and the block set is
\begin{equation}
\mathcal{B} = \{ D + v : v \in V \}, \quad \text{with } D = \{ x \in V : q(x) = 1 \}.
\end{equation}
Then, the automorphism group $\mathrm{Aut}(\mathcal{D})$ acts $2$-transitively on the point set $V$.
\end{lemma}

\subsection{Symplectic vector spaces and symplectic isomorphism}

In this subsection, we will introduce some concepts on symplectic vector spaces and symplectic isomorphism, see \cite{G} for the details.

\begin{definition}
A \emph{symplectic vector space} over a field $F$ is a pair $(V, \omega)$, where $V$ is a finite-dimensional $F$-vector space and $\omega : V \times V \to F$ is a non-degenerate, alternating bilinear form, i.e.,
\begin{enumerate}
    \item $\omega(u, v) = -\omega(v, u)$ for all $u, v \in V$,
    \item For every non-zero $u \in V$, there exists $v \in V$ such that $\omega(u, v) \neq 0$.
\end{enumerate}
\end{definition}

\begin{definition}
Let $(V, \omega_V)$ and $(W, \omega_W)$ be two symplectic vector spaces over the same field $F$. A vector space isomorphism $T : V \to W$ is called a \emph{symplectic isomorphism} (or \emph{symplectomorphism} in the linear setting) if it preserves the symplectic form:
\begin{equation}
\omega_W(T(u), T(v)) = \omega_V(u, v) \quad \text{for all } u, v \in V.
\end{equation}
If such an isomorphism $T$ exists, $(V, \omega_V)$ and $(W, \omega_W)$ are said to be \emph{symplectically isomorphic}, denoted $(V, \omega_V) \cong (W, \omega_W)$.
\end{definition}

\section{Proof of Theorem \ref{theorem1}}
{
Polujan and Pott introduced \cite[Example III.6]{PP} that there exist inequivalent bent functions whose translation designs are isomorphic, and then they extended this example to an infinite family in \cite[Theorem III.8]{PP}. We further extend the set of such examples via a special class of M-M bent functions (Proposition \ref{thm3.9}). Lemma \ref{thm:symplectic_2_transitive}, Proposition~\ref{thm3.9}, and  Kantor’s classification theorem of 2-transitive symmetric designs \cite{K2} yield a positive answer to Theorem~\ref{theorem1}.

}



\begin{lemma}\label{lemma5}
  {\rm For an M-M bent function $f(x_1,\ldots,x_n,y_1,\ldots,y_n)=f(x,y)=x\cdot y+ g(y)$ in $\mathcal{BF}_{2n}$, the following statements hold.
\begin{itemize}
    \item [(i)] The translation design $\operatorname{Dev}(D_f)$ is isomorphic to the translation designs $\operatorname{Dev}(D_{f_i})$ derived from $f_{i}(x,y)=f(x,y)+y_i$ for all $i$ with $1\le i\le n$.
    
    \item [(ii)] For $n\geq 2$, the translation design $\operatorname{Dev}(D_f)$ is isomorphic to the translation designs $\operatorname{Dev}(D_{f_{ij}})$ derived from $f_{ij}(x,y)=f(x,y)+y_iy_j$ for all $i,j$ with $1\le i<j\le n$.
    
    \item [(iii)] For $n\geq 3$, the translation design $\operatorname{Dev}(D_f)$ is isomorphic to the translation designs $\operatorname{Dev}(D_{f_{ijk}})$ derived from $f_{ijk}(x,y)=f(x,y)+y_iy_jy_k$ for all $i,j,k$ with $1\le i<j<k\le n$.
     
    \end{itemize}
 } 
  
\end{lemma}

\begin{proof}
    First, we prove (iii).  {Let $n\geq3$. Then for $1\le i<j<k\le n$, we can} define the map $\pi_{i,j,k}$ from $\mathbb{F}_2^{2n}$ to itself as
	\begin{align*}
		&\pi_{i,j,k}(x,y)=\pi_{i,j,k}(x_1,\ldots,x_n,y_1,\ldots,y_n)\\&=(x_1,\ldots,x_i+y_jy_k,\ldots,x_j+y_iy_k,\ldots,x_k+y_iy_j,\ldots,x_n,y_1,\ldots,y_n).
	\end{align*}
	{Since $\pi_{i,j,k}\circ\pi_{i,j,k}$ is an identity map, $\pi_{i,j,k}$ is a permutation of $\mathbb{F}^{2n}_2$.}
Then we have
       	\begin{align*}
		&f(\pi_{i,j,k}(x_1,\ldots,x_n,y_1,\ldots,y_n)+\pi_{i,j,k}(a_1,\ldots,a_n,b_1,\ldots,b_n))\\
		=&(x_1+a_1)(y_1+b_1)+\cdots+(x_i+y_jy_k+a_i+b_jb_k)(y_i+b_i)+\cdots\\
		&+(x_j+y_iy_k+a_j+b_ib_k)(y_j+b_j)+\cdots+(x_k+y_iy_j+a_k+b_ib_j)(y_k+b_k)+\cdots\\
		&+(x_n+a_n)(y_n+b_n)+g(y+b)\\
		=&(x_1+a_1)(y_1+b_1)+\cdots+(x_n+a_n)(y_n+b_n)\\
		&+(y_jy_k+b_jb_k)(y_i+b_i)+(y_iy_k+b_ib_k)(y_j+b_j)+(y_iy_j+b_ib_j)(y_k+b_k)+g(y+b)\\
		=&(x_1+a_1)(y_1+b_1)+\cdots+(x_n+a_n)(y_n+b_n)\\
		&+y_iy_jy_k+b_iy_jy_k+b_jy_iy_k+b_ky_iy_j+b_ib_jy_k+b_ib_ky_j+b_kb_jy_i+b_ib_jb_k+g(y+b)\\
        =& (x_1+a_1)(y_1+b_1)+\cdots+(x_n+a_n)(y_n+b_n)+(y_i+b_i)(y_j+b_j)(y_k+b_k)+g(y+b)\\
		=&f_{ijk}(x_1+a_1,\ldots,x_n+a_n,y_1+b_1,\ldots,y_n+b_n)\\
        =&f_{ijk}((x_1,\ldots,x_n,y_1,\ldots,y_n)+(a_1,\ldots,a_n,b_1,\ldots,b_n)),
	\end{align*}
	where $b=(b_1,\ldots,b_n)$. This proves (iii).  
    
    (ii). Let $n\geq 2$. For $1\le i<j\le n$, the map $\pi_{i,j}(x,y)=(x_1,\ldots,x_i+y_j,\ldots,x_n,y_1,\ldots,y_n)$ is a permutation of $\mathbb{F}^{2n}_2$. We have $f(\pi_{i,j}(x,y)+\pi_{i,j}(a,b))=f_{ij}((x,y)+(a,b))$, and the result follows.

    (i) Then for $1\le i\le n$, the map $\pi_i(x,y)=(x_1,\ldots,x_i+1,\ldots,x_n,y_1,\ldots,y_n)$ is a permutation of $\mathbb{F}^{2n}_2$. We can check that $f(\pi_{i}(x,y)+id(a,b))=f_i((x,y)+(a,b))$, and the result follows.  
    
        This completes the proof.
\end{proof}

{
\begin{proposition}\label{thm3.9}
{\rm
{Let $n\geq3$. For an M-M bent function $f(x_1,\ldots,x_n,y_1,\ldots,y_n)=f(x,y)=x\cdot y+ g(y)$ in $\mathcal{BF}_{2n}$, if $g$ has degree at most 3 and $g(\textbf{0})=0$ (resp., $g(\textbf{0})=1$), then the translation design $\operatorname{Dev}(D_f)$ is isomorphic to the translation design $\operatorname{Dev}(D_{x\cdot y})$ (resp., $\operatorname{Dev}(D_{x\cdot y+1})$). 

 } 
}
\end{proposition}

\begin{proof}
First, assume that $g$ is a cubic Boolean function with quadratic and linear terms. By applying Lemma \ref{lemma5} (iii) recursively, we can remove all cubic terms in $g$ while preserving the isomorphism, and then by applying Lemma \ref{lemma5} (ii) recursively, we can remove all quadratic terms in $g$ while preserving the isomorphism. Finally, by applying  Lemma \ref{lemma5} (i) recursively, we can remove all linear terms in $g$ while preserving the isomorphism. From these procedure, we obtain that $\operatorname{Dev}(D_f)$ is isomorphic to $\operatorname{Dev}(D_{x\cdot y})$ when $g(\textbf{0})=0$ and $\operatorname{Dev}(D_{x\cdot y+1})$ when $g(\textbf{0})=1$. This procedure is also valid for $g$ being any cubic, quadratic, or linear Boolean functions. 
\end{proof}
}

Now we are ready to prove Theorem \ref{theorem1}.

{\bf Proof of Theorem \ref{theorem1}:} As the condition  in Conjecture \ref{p1}, for a given function $f$, which is equivalent to an M-M bent function of the form $x\cdot y + g(y)$ with $\deg(g) \leq 3$, then the translation design $\operatorname{Dev}(D_f)$ is isomorphic to  $\operatorname{Dev}(D_{x\cdot y})$ or $\operatorname{Dev}(D_{x\cdot y+1})$ by Proposition \ref{thm3.9}.

By Kantor's classification theorem \cite{K2}, there is a unique symmetric \[2-(2^n, 2^{n-1} - 2^{n/2-1}, 2^{n-2} - 2^{n/2-1})\] design up to isomorphism with a $2$-transitive automorphism group, namely, the symplectic design ${S^{-}}(2n)$.  Thus its complementary design is unique up to isomorphism with a $2$-transitive automorphism group, namely, the symplectic design ${S^{+}}(2n)$. Moreover the classical symplectic translation designs $S^\pm(2n)$ are just generated by  quadratic bent functions $x\cdot y$ and $x\cdot y+1$. Then the result follows from Lemma \ref{thm:symplectic_2_transitive} and this completes the proof.

\section{Proof of Theorem \ref{thm:main-intro}}

This section is devoted to the proof of our second main theorem. Due to the length of the proof, we will divide into several subsections. Hereafter, set $[t] = \{1, 2, \ldots, t\}$ and $N=2^t-1$ for some integer $t\geq 2$.

%
\subsection{The vector space $I_t$}

In this subsection, we construct the algebraic foundation of our framework: a minimal finite-dimensional commutative square-zero algebra over $\mathbb{F}_2$.

Let $R_t = \mathbb{F}_2[u_1, \ldots, u_t] /  \langle u_1^2, \ldots, u_t^2\rangle$ be the quotient polynomial ring and let
\begin{equation}
I_t = \langle u_1, u_2, \ldots, u_t\rangle \subset R_t.
\end{equation}
It is easy to check that $I_t$ is the unique maximal ideal of $R_t$. For every non-empty index subset $S \subseteq [t]$, define the monomial basis elements as
\begin{equation}
e_S = \prod_{i \in S} u_i, \text{ where } e_{\emptyset}=1 \text{ by convention}.
\end{equation}
The ideal $I_t$ decomposes as a vector space over $\mathbb{F}_2$ in the following way.
\begin{equation}
I_t = \bigoplus_{\emptyset \neq S \subseteq [t]} \mathbb{F}_2 e_S.
\end{equation}
Then the dimension of $I_t$ over $\mathbb{F}_2$ is given by $\dim_{\mathbb{F}_2}(I_t) = 2^t - 1$.



The basic properties of $\mathbb{F}_2$-vector space $I_t$ are presented in the following proposition. 

\begin{proposition}\label{prop:mult-rules}
(1) For any non-empty subsets $S, T \subseteq [t]$, the multiplication of basis elements satisfies
\begin{equation}
e_S e_T =
\begin{cases}
e_{S \cup T}, & \text{if } S \cap T = \emptyset, \\
0, & \text{if } S \cap T \neq \emptyset.
\end{cases}
\end{equation}

Moreover, The ideal $I_t$ has the following properties:

(2) Every element $r \in I_t$ satisfies $r^2 = 0$.

(3) The nilpotency index of $I_t$, which is the smallest positive integer $k$ such that $I_t^k = 0$, is $t + 1$.
\end{proposition}
 
\begin{proof}
(1) If $S \cap T \neq \emptyset$, there exists an index $k \in S \cap T$, so that $e_S e_T$ contains $u_k^2 = 0$. If $S \cap T = \emptyset$, the products of distinct square-free variables combine directly to yield $e_{S \cup T}$.

(2) Let $r = \sum_{\emptyset \neq S \subseteq [t]} a_S e_S$ with $a_S \in \mathbb{F}_2$. Expanding $r^2$ over $\mathbb{F}_2$ leads to
\begin{equation}
r^2 = \sum_{S} a_S^2 e_S^2 + \sum_{S \neq T} 2 a_S a_T e_S e_T.
\end{equation}
By Proposition \ref{prop:mult-rules}, $e_S^2 = 0$ for all $S$. Then all cross terms vanish, establishing $r^2 = 0$.

(3) The product of all $t$ generators is $u_1 u_2 \cdots u_t = e_{[t]} \neq 0$, showing $I_t^t \neq 0$. Any product of $t + 1$ elements in $I_t$ contains at least one generator $u_i$ at least twice. By the square-zero relation $u_i^2 = 0$, we get $I_t^{t+1} = 0$.
\end{proof}

\begin{theorem}   
\label{lem:minimal-dimension}
Let $R$ be a finite-dimensional local commutative algebra over $\mathbb{F}_2$ and $r^2 = 0$ for all $r \in R$. If $R^t \neq 0$ for an integer $t \ge 1$, then 
$
\dim_{\mathbb{F}_2}(R) \ge 2^t - 1.
$
Furthermore, equality holds if and only if $R \cong I_t = \langle u_1, u_2, \ldots, u_t \rangle \subset \mathbb{F}_2[u_1, \ldots, u_t]/(u_1^2, \ldots, u_t^2)$ as algebras.
\end{theorem}

\begin{proof}
Since $R^t \neq 0$, there exist $t$ elements $u_1,  \ldots, u_t \in R$ such that their product $u_1 \cdots u_t \neq 0$. 

We claim that the $2^t - 1$ elements $\{ e_S : \emptyset \neq S \subseteq [t] \}$ are linearly independent over $\mathbb{F}_2$. Suppose $\sum_{\emptyset \neq S \subseteq [t]} c_S e_S = 0$ for some $c_S \in \mathbb{F}_2$. Let $S_0$ be a minimal subset (with respect to cardinality) such that $c_{S_0} \neq 0$. Multiplying the linear combination by $e_{[t] \setminus S_0}$, we obtain
\begin{equation}
c_{S_0} e_{S_0} e_{[t] \setminus S_0} + \sum_{S \neq S_0} c_S e_S e_{[t] \setminus S_0} = 0.
\end{equation}
For any $S \neq S_0$ with $c_S \neq 0$, we get the following:
\begin{enumerate}
    \item If $S \setminus S_0 \neq \emptyset$, that is, $S \cap ([t] \setminus S_0)\neq \emptyset$, then there exists one variable $u_k$ to appear twice in the product $e_S e_{[t] \setminus S_0}$.  Since $u_k^2 = 0$ for all $k$, we have $e_S e_{[t] \setminus S_0} = 0$.
    \item If $S \subset S_0$, then $|S| < |S_0|$, which contradicts the minimality of $S_0$ (so $c_S = 0$).
\end{enumerate}

Thus, the equation reduces to $c_{S_0} u_1 u_2 \cdots u_t = 0$. Since $u_1 u_2 \cdots u_t \neq 0$, it follows that $c_{S_0} = 0$, establishing linear independence. Hence, $\dim_{\mathbb{F}_2}(R) \ge 2^t - 1$.

If $\dim_{\mathbb{F}_2}(R) = 2^t - 1$, the set $\{e_S : \emptyset \neq S \subseteq [t]\}$ forms a vector space basis for $R$. The multiplication rules are forced to be $e_S e_T = e_{S \cup T}$ if $S \cap T = \emptyset$ and $0$ otherwise. This yields an explicit algebra isomorphism $R \cong I_t$.
\end{proof}




\subsection{Set-Partition Permutations }

In this subsection, we define the set-partition polynomial permutation $P_t$, construct $f_t$, and evaluate its algebraic degree.


Let $S$ be a nonempty finite set.
A \emph{set partition} of $S$ is a collection $\mathcal P=\{B_1,\ldots,B_r\}$ of nonempty subsets of $S$ such that
\[
B_i\cap B_j=\varnothing
\text{ for }i\neq j, \mbox{ and } \bigcup_{i=1}^{r}B_i=S.
\]
The subsets $B_1,\ldots,B_r$ are called the \emph{blocks} of $\mathcal P$. 
We denote by $\Pi(S)$ the set of all set partitions of $S$.

Now let $S \subseteq [t]$. For a vector 
\[
x = (x_T)_{\emptyset \neq T \subseteq [t]} \in \mathbb{F}_2^N,
\]
whose coordinates are indexed by the nonempty subsets of $[t]$, we define a Boolean function $P_S$ indexed by the set $S$ as follows:
\begin{equation}
P_S(x) = \sum_{\mathcal{P} \in \Pi(S)} \prod_{T \in \mathcal{P}} x_T.
\end{equation}
This yields the vectorial Boolean function $P_t = (P_S)_{\emptyset \neq S \subseteq [t]} : \mathbb{F}_2^N \to \mathbb{F}_2^N$.

\begin{definition}\label{deft}
Define the Boolean function $f_t : \mathbb{F}_2^N \times \mathbb{F}_2^N \to \mathbb{F}_2$ by
\begin{equation}
f_t(x, z) = z \cdot P_t(x) = \sum_{\emptyset \neq S \subseteq [t]} z_S P_S(x).
\end{equation}
\end{definition}

\begin{example}
    For $S=\{1,2,3\}$, we have 
    $$\Pi(S)=\{\{1,2,3\},\{\{1\},\{2,3\}\},\{\{2\},\{1,3\}\},\{\{3\},\{1,2\}\},\{\{1\},\{2\},\{3\}\}\},
    $$
     $$
    P_{\{1\}}(x)=x_{\{1\}},~ P_{\{2\}}(x)=x_{\{2\}},~ P_{\{3\}}(x)=x_{\{3\}},
    $$
    $$
P_{\{1,2\}}(x)=x_{\{1,2\}}+x_{\{1\}}x_{\{2\}},~P_{\{1,3\}}(x)=x_{\{1,3\}}+x_{\{1\}}x_{\{3\}},~P_{\{2,3\}}(x)=x_{\{2,3\}}+x_{\{2\}}x_{\{3\}},
$$
    $$
    P_S(x)=x_{\{1,2,3\}}+x_{\{1\}}x_{\{2,3\}}+x_{\{2\}}x_{\{1,3\}}+x_{\{3\}}x_{\{1,2\}}+x_{\{1\}}x_{\{2\}}x_{\{3\}}.
    $$
    Thus $$
    P_3(x)=(P_{\{1\}}(x),P_{\{2\}}(x),P_{\{3\}}(x),P_{\{1,2\}}(x),P_{\{1,3\}}(x),P_{\{2,3\}}(x),P_S(x))
    $$
    and so  
    \begin{multline*}
        f_3(x,z)=x_{\{1\}}z_{\{1\}}+x_{\{2\}}z_{\{2\}}+x_{\{3\}}z_{\{3\}}+
    (x_{\{1,2\}}+x_{\{1\}}x_{\{2\}})z_{\{1,2\}}
    +(x_{\{1,3\}}+x_{\{1\}}x_{\{3\}})z_{\{1,3\}}\\+(x_{\{2,3\}}+x_{\{2\}}x_{\{3\}})z_{\{2,3\}}+
    (x_{\{1,2,3\}}+x_{\{1\}}x_{\{2,3\}}+x_{\{2\}}x_{\{1,3\}}+x_{\{3\}}x_{\{1,2\}}+x_{\{1\}}x_{\{2\}}x_{\{3\}})z_{\{1,2,3\}},
    \end{multline*}
   which is a Boolean function of $\mathbb{F}^7_2\times \mathbb{F}^7_2$ whose algebraic degree is four.
\end{example}

\subsection{Symplectic Isomorphism}

In this subsection, we introduce the direct sum $I_t \oplus  \operatorname{Hom}_{\mathbb{F}_2}(I_t, \mathbb{F}_2)$ defined with $(10)$ and prove that $f_t$ defined in Definition \ref{deft} induces a classical symplectic translation design.


Recall from Theorem \ref{lem:minimal-dimension} that $I_t$ is a finite-dimensional $\mathbb{F}_2$-algebra with basis $\{e_S : \emptyset \neq S \subseteq [t]\}$. Let $I_t^* = \operatorname{Hom}_{\mathbb{F}_2}(I_t, \mathbb{F}_2)$ be its dual vector space and let $\{e_S^* : \emptyset \neq S \subseteq [t]\}$ be the corresponding dual basis of $I_t^*$, satisfying
\[
e_S^*(e_T) = \delta_{S,T}
\]
for all nonempty $S, T \subseteq [t]$, where $\delta_{S,T}$ denotes the Kronecker delta.

Consider the vector space direct sum
$ I_t \oplus I_t^*.
$
Define multiplication on $I_t \oplus I_t^*$ by
\begin{equation}
(r, \alpha)(s, \beta) = (rs, r\beta + s\alpha),
\end{equation}
where $(r\beta)(x) = \beta(rx)$ for all $x \in I_t$. Define the circle operation $\circ$ on $I_t \oplus I_t^*$ as
\begin{equation}
a \circ b = a + b + ab.
\end{equation}

\begin{proposition}
Every element $a \in I_t \oplus I_t^*$ satisfies $a^2 = 0$. Consequently, $(I_t \oplus I_t^*, \circ)$ is an elementary abelian $2$-group of order $2^{2N}$.
\end{proposition}

\begin{proof}
For any $a = (r, \alpha) \in I_t \oplus I_t^*$, $a^2 = (r^2, 2r\alpha) = (0, 0)$ because $r^2 = 0$ in $I_t$ by Proposition \ref{prop:mult-rules}. Associativity follows from $1 + (a \circ b) = (1 + a)(1 + b)$. Finally, $a \circ a = 2a + a^2 = 0$, which implies every non-zero element has order $2$.
\end{proof}


Define the canonical quadratic form $q : I_t \oplus I_t^* \to \mathbb{F}_2$ by
\begin{equation}
q(r, \alpha) = \alpha(r).
\end{equation}
Its polar form $B(a, b) = q(a+b) + q(a) + q(b)$ is an alternating, non-degenerate symplectic bilinear form on $I_t \oplus I_t^*$, see \cite{G}. For $x, z \in \mathbb{F}_2^N$, define $r : \mathbb{F}_2^N \to I_t$ and $\beta : \mathbb{F}_2^N \to I_t^*$ respectively by 
\begin{equation}
r(x) = \sum_{\emptyset \neq S \subseteq [t]} P_S(x) e_S, \quad \beta(z) = \sum_{\emptyset \neq S \subseteq [t]} z_S e_S^*.
\end{equation}

\begin{definition}
Define the coordinate transformation map $\Phi_t : \mathbb{F}_2^{2N} \to I_t \oplus I_t^*$ by
\begin{equation}
\Phi_t(x, z) = \bigl(r(x), \beta(z) + r(x)\beta(z)\bigr).
\end{equation}
\end{definition}

\begin{theorem}\label{thm:isomorphism}
The mapping $\Phi_t : (\mathbb{F}_2^{2N}, +) \to (I_t \oplus I_t^*, \circ)$ is a group isomorphism.
\end{theorem}

\begin{proof}
By the definition of group isomorphism, first we need to prove that $$\Phi_t((x, z) + (y, w)) = \Phi_t(x, z) \circ \Phi_t(y, w)$$ for any $(x, z), (y, w) \in \mathbb{F}_2^{2N}$. For brevity, write
$
r_x=r(x), \beta_z=\beta(z),
$
and set
$
\alpha_z=\beta_z+r_x\beta_z=(1+r_x)\beta_z.
$
Then
$
\Phi_t(x,z)=(r_x,\alpha_z).
$ We first record two identities. Since \(\beta\) is linear,
$
\beta_{z+w}=\beta_z+\beta_w.
$
Also, expanding the defining set-partition sums for \(P_S(x+y)\) gives
$
r_{x+y}=r_x+r_y+r_xr_y.
$
Indeed, the mixed terms in the expansion are exactly the terms contributing to \(r_xr_y\).

Now compute the circle product. Since
\[
(r,\alpha)\circ(s,\gamma)=(r+s+rs,\ \alpha+\gamma+r\gamma+s\alpha),
\]
we obtain
\[
\begin{aligned}
\Phi_t(x,z)\circ\Phi_t(y,w)
&=(r_x,\alpha_z)\circ(r_y,\alpha_w)\\
&=(r_x+r_y+r_xr_y,\ \alpha_z+\alpha_w+r_x\alpha_w+r_y\alpha_z)\\
&=(r_{x+y},\ \alpha_z+\alpha_w+r_x\alpha_w+r_y\alpha_z).
\end{aligned}
\]
For the second component, using \(\alpha_z=(1+r_x)\beta_z\) and
\(\alpha_w=(1+r_y)\beta_w\), we get
\[
\begin{aligned}
&\alpha_z+\alpha_w+r_x\alpha_w+r_y\alpha_z\\
&=(1+r_x)\beta_z+(1+r_y)\beta_w
  +r_x(1+r_y)\beta_w+r_y(1+r_x)\beta_z\\
&=(1+r_x)(1+r_y)\beta_z+(1+r_x)(1+r_y)\beta_w\\
&=(1+r_x)(1+r_y)(\beta_z+\beta_w)\\
&=(1+r_x+r_y+r_xr_y)\beta_{z+w}\\
&=(1+r_{x+y})\beta_{z+w}=\beta_{z+w}+r_{x+y}\beta_{z+w}.
\end{aligned}
\]
Therefore
\[
\Phi_t(x,z)\circ\Phi_t(y,w)
=
\bigl(r_{x+y},\ \beta_{z+w}+r_{x+y}\beta_{z+w}\bigr)
=
\Phi_t(x+y,z+w).
\]
Thus \(\Phi_t\) is a group homomorphism.

It remains to prove that \(\Phi_t\) is bijective. The map
\[
x\mapsto r_x=\sum_{\emptyset\neq S\subseteq[t]}P_S(x)e_S
\]
is bijective. Indeed, in the basis \(\{e_S\}\), the coordinate \(P_S(x)\) equals \(x_S\) plus a polynomial depending only on variables \(x_T\) with \(|T|<|S|\). Hence \(P_t\) is lower triangular with unit diagonal, so it is invertible over \(\mathbb F_2\). Therefore \(x\) is uniquely determined by \(r_x\).

Now fix \(x\). The map
$
z\mapsto \beta_z
$
is a linear isomorphism, and \(1+r_x\) is invertible in \(I_t\) because
$
(1+r_x)^2=1+r_x^2=1
$
by Proposition~\ref{prop:mult-rules}(2). Hence
$
z\mapsto (1+r_x)\beta_z=\beta_z+r_x\beta_z
$
is also a bijection. Consequently, for each
$
(r,\alpha)\in I_t\oplus I_t^*,
$
there is a unique \(x\) with \(r_x=r\), and then a unique \(z\) with
$
(1+r_x)\beta_z=\alpha.
$
Thus \(\Phi_t\) is bijective.

Therefore \(\Phi_t\) is a group isomorphism  and this completes the proof.
\end{proof}


\begin{theorem}\label{thm:pullback}
For all $(x, z) \in \mathbb{F}_2^{2N}$, we have $q(\Phi_t(x, z)) = f_t(x, z)$.
\end{theorem}

\begin{proof}
Evaluating $q$ on $\Phi_t(x, z) = (r, \beta + r\beta)$, we get
\begin{equation}
q(\Phi_t(x, z)) = (\beta + r\beta)(r) = \beta(r) + \beta(r^2).
\end{equation}
Observe that $\beta(r^2) = 0$ since $r^2 = 0$ for all $r \in I_t$. Expanding $\beta(r)$ in coordinates leads to
\begin{equation}
\beta(r) = \left(\sum_{\emptyset \neq S} z_S e_S^*\right) \left(\sum_{\emptyset \neq T} P_T(x) e_T\right) = \sum_{\emptyset \neq S} z_S P_S(x) = f_t(x, z).
\end{equation}
\end{proof}

\subsection{The proof}

Now we are ready to prove Theorem \ref{thm:main-intro}.

{\bf Proof of Theorem \ref{thm:main-intro}:}

\begin{enumerate}

\item  We prove that for every integer $t \ge 2$, $f_t$ is a Maiorana--McFarland bent function on $\mathbb{F}_2^{2N}$.

The single-block partition $\mathcal{P} = \{S\}$ gives the term $x_S$. Thus, $P_S(x) = x_S + Q_S(x)$, where $Q_S(x)$ depends strictly on variables $x_T$ with $|T| < |S|$. Thus, the system $y = P_t(x)$ is lower triangular with unit diagonal, making $P_t$ an invertible polynomial mapping (a permutation on $\mathbb{F}_2^N$). By Dillon's Maiorana--McFarland theorem \cite{D}, $f_t(x, z) = z \cdot P_t(x)$ is bent.

\item We prove that for every integer $t \ge 2$, the algebraic degree of $f_t$ is precisely
$
\deg(f_t) = t + 1.
$

Since $\deg(P_S) \le |S|$, we have $\deg(z_S P_S(x)) \le |S| + 1 \le t + 1$, so $\deg(f_t) \le t + 1$.

Now inspect the term indexed by $S = [t]$. The partition $\mathcal{P}_0 = \{\{1\}, \{2\}, \ldots, \{t\}\} \in \Pi([t])$ produces the unique monomial $\prod_{i=1}^t x_{\{i\}}$. Thus, $f_t(x, z)$ contains the monomial
\begin{equation}
M^* = z_{[t]} \prod_{i=1}^t x_{\{i\}}.
\end{equation}
Because $M^*$ uniquely contains $z_{[t]}$, it cannot be produced by any other term $z_S P_S(x)$ for $S \neq [t]$. Consequently, $M^*$ experiences no algebraic cancellation in the ANF of $f_t$. Since $\deg(M^*) = 1 + t$, we conclude $\deg(f_t) = t + 1$.

\item  We prove that the translation design $\operatorname{Dev}(D_{f_t})$ is isomorphic to  $S^\pm(2(2^t - 1))$.

Let $D_{f_t} = \{ v \in \mathbb{F}_2^{2N} : f_t(v) = 1 \}$ and $D_q = \{ a \in I_t \oplus I_t^* : q(a) = 1 \}$. By Theorem \ref{thm:pullback}, we get $D_{f_t} = \Phi_t^{-1}(D_q)$, and by Theorem \ref{thm:isomorphism}, $\Phi_t$ is a group isomorphism from $(\mathbb{F}_2^{2N}, +)$ to $(I_t \oplus I_t^*, \circ)$. Thus, $\Phi_t$ maps translates of $D_{f_t}$ in $(\mathbb{F}_2^{2N}, +)$ directly to circle translates of $D_q$ in $(I_t \oplus I_t^*, \circ)$. Since $(I_t \oplus I_t^*, \circ, q)$ defines the canonical symplectic design $S^\pm(2N)$, we conclude $\operatorname{Dev}(D_{f_t}) \cong S^\pm(2(2^t - 1))$.

\end{enumerate}
This completes the proof.
\section{Conclusion}

In this paper, we have resolved a conjecture of Polujan and Pott. By combining the minimal algebra  with set-partition polynomial permutations, we constructed an infinite family of Maiorana--McFarland bent functions $f_t$ in $2(2^t-1)$ variables whose algebraic degree $\deg(f_t) = t + 1$ grows without bound. Furthermore, by uncovering an elementary abelian group structure within the  direct sum $  I_t \oplus I_t^*$, we proved that the translation designs generated by $f_t$ are isomorphic to classical symplectic designs $S^\pm(2(2^t-1))$.

\section*{Declaration on the Use of Generative AI}

During the preparation of this manuscript, ChatGPT 5.6 Sol was used exclusively in Section IV, where it assisted with problem-solving as well as the verification and revision of the results. All other sections, including their statements and proofs, were developed entirely by the authors. The authors reviewed and approved all AI-assisted content in Section IV and made the final decisions regarding the statements, proofs, organization, and conclusions of the paper. The authors take full responsibility for the correctness and content of the manuscript.

\end{document}